\documentclass[runningheads]{llncs}

\usepackage{times}
\usepackage{soul}
\usepackage{url}

\usepackage[utf8]{inputenc}
\usepackage[small]{caption}
\usepackage{graphicx}
\usepackage{subcaption}
\usepackage{amssymb}
\usepackage{amsfonts}
\usepackage{amsmath}

\usepackage{amsthm}
\usepackage{algorithm}
\usepackage{algorithmic}
\usepackage{cite}
\usepackage[american]{babel}
\usepackage{bm}
\usepackage{mathtools}
\usepackage{pifont}
\usepackage{booktabs} 
\usepackage{siunitx} 
\newcommand{\bmx}{\bm{x}}

\newcommand{\bmz}{\bm{z}}

\renewcommand{\tilde}{\widetilde}
\renewcommand{\hat}{\widehat}

\usepackage{multirow}
\usepackage{hyperref}

\begin{document}
\title{Biased Pareto Optimization for Subset Selection \\with Dynamic Cost Constraints\thanks{This work was supported by the National Science and Technology Major Project (2022ZD0116600) and National Science Foundation of China (62276124). Chao Qian is the corresponding author.}}
%
%

\author{Dan-Xuan Liu \orcidID{0000-0002-7076-823X} \and Chao Qian\orcidID{0000-0001-6011-2512} }
\authorrunning{Dan-Xuan Liu and Chao Qian}
%
\institute{National Key Laboratory for Novel Software Technology, Nanjing University, China \\ School of Artificial Intelligence, Nanjing University, China \\
\email{\{liudx,qianc\}@lamda.nju.edu.cn}}

\maketitle              
\begin{abstract}
Subset selection with cost constraints aims to select a subset from a ground set to maximize a monotone objective function without exceeding a given budget, which has various applications such as influence maximization and maximum coverage. In real-world scenarios, the budget, representing available resources, may change over time, which requires that algorithms must adapt quickly to new budgets. However, in this dynamic environment, previous algorithms either lack theoretical guarantees or require a long running time. The state-of-the-art algorithm, POMC, is a Pareto optimization approach designed for static problems, lacking consideration for dynamic problems. In this paper, we propose BPODC, enhancing POMC with biased selection and warm-up strategies tailored for dynamic environments. We focus on the ability of BPODC to leverage existing computational results while adapting to budget changes. We prove that BPODC can maintain the best known $(\alpha_f/2)(1-e^{-\alpha_f})$-approximation guarantee when the budget changes. Experiments on influence maximization and maximum coverage show that BPODC adapts more effectively and rapidly to budget changes, with a running time that is less than that of the static greedy algorithm.

\end{abstract}
\section{Introduction}\label{introduction}

The subset selection problem is a general NP-hard problem, which has a wide range of applications, such as influence maximization~\cite{IM-application}, maximum coverage~\cite{MC-application} and sensor placement~\cite{SP-application}, to name a few. The goal is to select a subset $X$ from a ground set $V$ to optimize a given function $f$, subject to a cost constraint. Specifically, the cost of the subset $X$ must not exceed a given budget $B$. This can be formally expressed as follows:
\begin{equation}\label{eq-general-problem}
\mathop{\arg\max}\limits_{X \subseteq V}  f(X) \quad \text{s.t.} \quad c(X) \leq B,
\end{equation}
where the objective function $f: 2^V \!\rightarrow\! \mathbb{R}$ and the cost function $c: 2^V \!\rightarrow\! \mathbb{R}$ are both monotone, meaning they do not decrease with the addition of elements to the set $X$. However, these functions are not necessarily submodular; a set function $f$ is submodular if $\forall X \subseteq Y, v \notin Y: f(X \cup \{v\})\!-\!f(X) \geq f(Y \cup \{v\})\! -\! f(Y),$ implying the diminishing returns property~\cite{nemhauser1978analysis}. We will introduce two applications of this problem in the next section. For the general constraint presented in Eq.~(\ref{eq-general-problem}), the Generalized Greedy Algorithm (GGA) iteratively selects an item with the largest ratio of marginal gain on $f$ and cost $c$, achieving the best known $(\alpha_f/2)(1-e^{-\alpha_f})$-approximation ratio~\cite{POMC,GGA}, where $\alpha_f$ measures the degree to which $f$ is nearly submodular.

In real-world scenarios, the resource budget $B$ in Eq.~(\ref{eq-general-problem}) for subset selection problems may vary over time. For example, in influence maximization, the market investment budget changes based on company outcomes and strategies, leading to a dynamic cost constraint. After each change of the budget $B$, it is feasible to treat the problem as a static problem with the new budget, and run static algorithms (e.g., GGA) from scratch, which, however, may lead to a long running time. When problem constraints change frequently, static algorithms might not be able to provide new solutions in time. For subset selection with dynamic cost constraints, Roostapour \textit{et al.}~\cite{aij22Roostapour} proposed an Adaptive Generalized Greedy Algorithm named AGGA, adjusting the current solution to fit the new budget. If the budget is reduced, it removes items with smallest ratio of the marginal gain on $f$ and cost $c$. Conversely, if the budget increases, it adds items like the static GGA. However, AGGA cannot maintain an $(\alpha_f/2)(1-e^{-\alpha_f})$-approximation for the new budget and may even perform arbitrarily poorly~\cite{aij22Roostapour}. 

Unlike memoryless deterministic greedy algorithms, Evolutionary Algorithms (EAs) are capable of leveraging existing computational results while adapting to changes in the budget $B$. EAs can continually search for new solutions by using the genetic information from the parent individuals in the population. Qian \textit{et al.}~\cite{POMC} proposed a Pareto Optimization approach for maximizing a Monotone function with a monotone Cost constraint, called POMC~\cite{ecj15submodular,qian2015subset}. It reformulates the original subset selection problem as a bi-objective optimization problem, which maximizes the objective $f$ and minimizes the cost $c$ simultaneously, and then employs a multi-objective EA to solve it. For the static setting of a fixed budget $B$, POMC achieves the best known $(\alpha_f/2)(1-e^{-\alpha_f})$-approximation ratio using at most $O(nBP_{max}/\delta_{\hat{c}})$ expected running time\footnote{The expected running time refers to the expected number of evaluations of the objective function, as evaluations are usually the most expensive part of the process.}, where $n$ is the size of the ground set $V$, $P_{max}$ is the largest size of population during the run of POMC, $\hat{c}$ is an approximation of the cost function $c$ in case the exact computation of $c$ is impractical in real-world scenarios, and $\delta_{\hat{c}}=\min\{\hat{c}(X\cup v)-\hat{c}(X)|X\subseteq V, v\notin X\}$. When the budget $B$ decreases, POMC has already achieved the $(\alpha_f/2)(1-e^{-\alpha_f})$-approximation ratio; when the budget $B$ increases to $B'$, POMC can regain the $(\alpha_f/2)(1-e^{-\alpha_f})$-approximation ratio using at most $O(n\Delta_BP_{max}/\delta_{\hat{c}})$ expected running time~\cite{aij22Roostapour}, where $\Delta_B=B'-B$. However, the running time of POMC may not be polynomial because the population size $P_{max}$ can grow exponentially~\cite{aij22Roostapour}. Bian et al.~\cite{EAMC} proposed a single-objective EA for maximizing a Monotone function with a monotone Cost constraint, called EAMC, which maximizes a surrogate objective $g(X) = f(X)/(1 - e^{-\alpha_f \hat{c}(X) / B})$, and maintains at most two solutions for each possible size in the population. For the static setting, EAMC ensures the best known approximation ratio using at most $O(n^3)$ expected running time, which is polynomial. However, the surrogate objective $g$ of EAMC changes with the budget $B$, rendering previous solutions potentially ineffective for the new $B$. Bian et al.~\cite{fpomc} then proposed the Fast Pareto Optimization algorithm for maximizing a Monotone function with a monotone Cost constraint, called FPOMC, which is modified from POMC by introducing a greedy selection strategy and estimating the goodness of a solution to be selected. For the static setting, FPOMC obtains the best known $(\alpha_f/2)(1-e^{-\alpha_f})$-approximation ratio using at most $O(n^2K_{B})$ expected running time, which is also polynomial, where $K_{B}$ denotes the largest size of a subset satisfying the constraint $c(X)\leq B$. If the budget decreases, FPOMC maintains the best known approximation ratio; if the budget increases to $B'$, FPOMC can regain the same ratio within an expected running time of $O(n K_{B'}(K_{B'}-K_{B}))$, where $K_{B'}$ is the maximum subset size satisfying the new budget constraint $c(X)\leq B'$. We summarize the related works in Table~\ref{raleted-works}.

Among the algorithms suitable for dynamic environments (AGGA, POMC, EAMC and FPOMC), POMC empirically performs the best~\cite{aij22Roostapour}. However, it was originally designed for static problems and lacks considerations for dynamic environments. A natural question is whether it is possible to design a more advanced algorithm for dynamic environments that can quickly adapt its solutions when the budget $B$ changes. If so, how fast can this algorithm adapt, and can it surpass the speed of the static GGA? 

\begin{table}[t!]
\centering
\caption{Summary of algorithms with approximation guarantees and running time for the subset selection problem with static and dynamic cost constraints. The results obtained in this paper are highlighted in boldface. A ``-" denotes the algorithm is not suitable for the corresponding case, while a ``\ding{55}" means the algorithm is applicable but its performance is unknown.}\label{raleted-works}
\resizebox{\textwidth}{!}{\begin{tabular}{@{}l|c|c|c|c@{}}
\toprule
\multirow{2}{*}{\textbf{Algorithm}} & \multicolumn{2}{c|}{\textbf{Static}} & \multicolumn{2}{c}{\textbf{Dynamic}} \\\cline{2-5}
& \textbf{Guarantee} & \textbf{Running time} & \textbf{Guarantee} & \textbf{Running time} \\
\midrule
\multicolumn{1}{l|}{GGA~\cite{GGA}}    & $(\alpha_f/2)(1-e^{-\alpha_f})$ & $O(n^2)$ & - & - \\
\multicolumn{1}{l|}{AGGA~\cite{aij22Roostapour}}  & - & - & \ding{55} & $O(n^2)$ \\
\multicolumn{1}{l|}{POMC~\cite{POMC,aij22Roostapour}}  & $(\alpha_f/2)(1-e^{-\alpha_f})$ & $O(nBP_{max}/\delta_{\hat{c}})$ & $(\alpha_f/2)(1-e^{-\alpha_f})$ & $O(n\Delta_BP_{max}/\delta_{\hat{c}})$ \\
\multicolumn{1}{l|}{EAMC~\cite{EAMC}}  & $(\alpha_f/2)(1-e^{-\alpha_f})$ & $O(n^3)$ & \ding{55} & \ding{55} \\
\multicolumn{1}{l|}{FPOMC~\cite{fpomc}} & $(\alpha_f/2)(1-e^{-\alpha_f})$ & $O(n^2 K_{B})$   & $(\alpha_f/2)(1-e^{-\alpha_f})$  & $O(n K_{B'}(K_{B'}-K_{B}))$\\
\multicolumn{1}{l|}{\textbf{BPODC}} & \boldmath{$(\alpha_f/2)(1-e^{-\alpha_f})$} & \boldmath{$O(nB/(p_{min}\delta_{\hat{c}}))$} & \boldmath{$(\alpha_f/2)(1-e^{-\alpha_f})$} & \boldmath{$O(n\Delta_B/(p_{min}\delta_{\hat{c}}))$} \\
\bottomrule
\end{tabular}}
\end{table}

In this paper, we introduce BPODC based on Biased Pareto Optimization for maximizing a monotone function under Dynamic Cost constraints. BPODC is a modification of the POMC algorithm, enhanced with a biased selection mechanism and a warm-up strategy. POMC selects a solution uniformly at random from the population, often resulting in the choice of an ``unnecessary" solution and leading to inefficiency. In contrast, BPODC employs a biased selection strategy that favors solutions with cost values closer to the current budget $B$, which are selected with a higher probability. During the initial phase of one run, BPODC employs uniform selection temporarily as a warm-up to obtain a diverse population, which will lead the biased selection strategy to be more efficient. As in~\cite{dynamicgraphcoloring,DynamicPartitionMatroid}, we focus on the ability of BPODC to leverage existing computational results while adapting to budget changes. We prove that BPODC maintains the best-known $(\alpha_f/2)(1-e^{-\alpha_f})$-approximation ratio when the budget $B$ decreases. When the budget increases to $B'$, BPODC can regain the same ratio within an expected running time of $O(n\Delta_B/(p_{min}\delta_{\hat{c}}))$, where $\Delta_B = B' - B$, and $p_{min}$ is the minimum probability of selecting a solution during the run of BPODC. Through empirical evaluation on the applications of influence maximization and maximum coverage, we show that BPODC can find better solutions than previous algorithms while requiring less running time than the static greedy algorithm GGA, offering an alternative for solving subset selection problems with dynamic cost constraints.

\section{Subset Selection with Cost Constraints}

Let $\mathbb{R}$ and $\mathbb{R}^+$ denote the set of reals and non-negative reals, respectively. The set $V=\{v_1,v_2,\ldots,v_n\}$ denotes a ground set. A set function $f:2^V \rightarrow \mathbb{R}$ is monotone if $\forall X \subseteq Y: f(X) \leq f(Y)$. A set function $f$ is submodular if $\forall X \subseteq Y, v \notin Y: f(X \cup \{v\})-f(X) \geq f(Y \cup \{v\}) - f(Y)$, which intuitively represents the diminishing returns property~\cite{nemhauser1978analysis}, i.e., adding an item to a set $X$ gives a larger benefit than adding the same item to a superset $Y$ of $X$. The submodularity ratio in Definition~\ref{def-submodularity-ratio} characterizes how close a set function is to submodularity. When $f$ satisfies the monotone property, we have $0 \leq \alpha_f \leq 1$, and $f$ is submodular iff $\alpha_f=1$. 

\begin{definition}[Submodularity Ratio~\cite{GGA,DBLP:conf/ijcai/Qian0T18}]\label{def-submodularity-ratio}
The submodularity ratio of a non-negative set function $f$ is defined as $\alpha_f=\min_{X \subseteq Y,v \notin Y} \frac{f(X \cup \{v\})-f(X)}{f(Y \cup \{v\})-f(Y)}.$
\end{definition}

As presented in Definition~\ref{def:static-problem}, the subset selection problem with static cost constraints is to maximize a monotone objective function $f$ such that a monotone cost function $c$ is no larger than a budget $B$. We assume w.l.o.g. that monotone functions are normalized, i.e., $f(\emptyset)=0$ and $c(\emptyset)=0$. Since the exact computation of $c(X)$ may be unsolvable in polynomial time in some real-world applications~\cite{POMC,GGA}, we assume that only an $\psi(n)$-approximation $\hat{c}$ can be obtained, where $\forall X \subseteq V: c(X) \leq \hat{c}(X) \leq \psi(n) \cdot c(X)$. If $\psi(n)=1$, $\hat{c}(X)=c(X)$.

\begin{definition}[Subset Selection with Static Cost Constraints]\label{def:static-problem}
Given a monotone objective function $f: 2^V \rightarrow \mathbb{R}^+$, a monotone cost function $c: 2^V \rightarrow \mathbb{R}^+$ and a budget $B$, to find
\begin{align}\label{eq:general-def}
\mathop{\arg\max}\nolimits_{X \subseteq V} f(X) \quad \text{s.t.}\quad c(X)\leq B .
\end{align}
\end{definition}

The static problem in Definition~\ref{def:static-problem} assumes that the budget $B$ is fixed. However, in real-world scenarios, the resources often change over time, and thus the budget $B$ may change dynamically. For example, in influence maximization, the market investment budget changes based on company outcomes and strategies, leading to dynamic cost constraints, that is, the budget $B$ in Eq.~(\ref{eq:general-def}) may change over time. In this paper, we focus on the subset selection problem with dynamic cost constraints, given in Definition~\ref{def:dynamic-problem}. Whenever the budget $B$ changes, the problem can be treated as a new static problem using the updated budget, and static algorithms can be applied from the beginning. However, this may lead to a long running time.


\begin{definition}[Subset Selection with Dynamic Cost Constraints]\label{def:dynamic-problem}
Given a monotone objective function $f: 2^V \!\rightarrow\! \mathbb{R}^+$, a monotone cost function $c: 2^V \!\rightarrow\! \mathbb{R}^+$, and a sequence of changes on the budget $B$, to find a subset optimizing Eq.~(\ref{eq:general-def}) for each new $B$.
\end{definition}

\noindent\textbf{Influence Maximization.} Influence maximization in Definition~\ref{def-IM} is to identify a set of influential users in social networks~\cite{IM-application}.  A social network can be represented by a directed graph $G=(V,E)$, where each node represents a user and each edge $(u,v) \in E$ has a probability $p_{u,v}$ representing the influence strength from user $u$ to $v$. Given a budget $B$, influence maximization is to find a subset $X$ of $V$ such that the expected number of nodes activated by propagating from $X$ is maximized, while not violating the cost constraint $c(X) \leq B$. Here we use the fundamental propagation model called Independence Cascade (IC). Starting from a seed set $X$, it uses a set $A_t$ to record the nodes activated at time $t$, and at time $t+1$, each inactive neighbor $v$ of $u\in A_t$ becomes active with probability $p_{u,v}$; this process is repeated until no nodes get activated at some time. The set of nodes activated by propagating from $X$ is denoted as $IC(X)$, which is a random variable. The objective $\mathbb{E}[IC(X)]$ denotes the expected number of nodes activated by propagating from $X$, which is monotone and submodular.

\begin{definition}[Influence Maximization]\label{def-IM}
Given a directed graph $G=(V,E)$, edge probabilities $p_{u,v}$ where $(u,v) \!\in\! E$, a monotone cost function $c:2^V \rightarrow\mathbb{R}^+$ and a budget $B$, to find
\begin{align*}
\mathop{\arg\max}\nolimits_{X \subseteq V} \mathbb{E}[|IC(X)|] \quad \text{s.t.}\quad c(X)\leq B.
\end{align*}
\end{definition}

\noindent\textbf{Maximum Coverage.}
Given a family of sets that cover a universe of elements, maximum coverage as presented in Definition~\ref{def:MC} is to select some sets whose union is maximal under a cost budget. It is easy to verify that $f$ is monotone and submodular.

\begin{definition}[Maximum Coverage]\label{def:MC}
Given a set $U$ of elements, a collection $V=\{S_1,S_2,\ldots,S_n\}$ of subsets of $U$, a
monotone cost function $c: 2^V\rightarrow \mathbb{R}^+$ and a budget $B$, to find
\begin{align*}
\mathop{\arg\max}\nolimits_{X \subseteq V} f(X)=|\bigcup\nolimits_{S_i\in X} S_i| \quad \text{ s.t. } \quad c(X)\leq B.
\end{align*}
\end{definition}

\subsection{Previous Algorithms}
We now introduce five algorithms capable of solving the subset selection problem with dynamic cost constraints in Definition~\ref{def:dynamic-problem}.

\textbf{GGA.} The Generalized Greedy Algorithm (GGA) proposed in~\cite{GGA} selects one item maximizing the ratio of the marginal gain on $f$ and $\hat{c}$ in each iteration. After examining all items, the found subset is compared with the best single item (i.e., $v^*\in \arg\max_{v\in V: \hat{c}(\{v\})\leq B} f(\{v\})$), and the better one is returned. Let 
\begin{equation}\label{eq-optimal}
f(\tilde{X})=\max\left\{f(X) \mid
c(X) \leq  B\cdot \frac{\alpha_{\hat{c}}(1 + \alpha^2_{c}(K_B - 1)(1 - \kappa_c))}{\psi(n)K_B}\right\},
\end{equation}
where $\alpha_c$ and $\alpha_{\hat{c}}$ are the submodularity ratios of the cost function $c$ and its approximation $\hat{c}$, respectively, $\kappa_c=1-\min _{v \in V: c(\{v\})>0} \frac{c(V)-c(V \backslash\{v\})}{c(\{v\})}$ is the total curvature of $c$, and $K_B=\max\{|X| \mid c(X) \leq B\}$, i.e., the largest size of a subset satisfying the constraint. As $1-\kappa_c\le 1/\alpha_c$, $0 \leq \alpha_{\hat{c}}, \alpha_{c} \leq 1$ and $\psi(n) \geq 1$, it holds that $$\frac{\alpha_{\hat{c}}(1 + \alpha^2_{c}(K_B - 1)(1 - \kappa_c))}{\psi(n)K_B} \leq 1.$$ Thus, $\tilde{X}$ is actually an optimal solution of Eq.~(\refeq{eq:general-def}) with a slightly smaller budget constraint. GGA can solve only the static problem in Definition~\ref{def:static-problem}. Qian \textit{et al.}~\cite{POMC} proved that GGA can obtain a subset $X$ satisfying $f(X) \geq  (\alpha_f/2)\cdot (1-e^{-\alpha_f})\cdot f(\tilde{X})$. The dynamic problem in Definition~\ref{def:dynamic-problem} can be viewed as a series of static problems, each with a different budget $B$. After the budget changes, GGA can be applied to solve the subsequent static problem with the updated budget from scratch.

\textbf{AGGA.} To face the dynamic changes of $B$, Roostapour~\textit{et al.}~\cite{aij22Roostapour} introduced a natural Adaptive version of the static GGA, named AGGA. When $B$ increases, AGGA continues to iteratively add an item with the largest ratio of the marginal gain on $f$ and cost $c$ to the current solution, just like GGA; when $B$ decreases, it iteratively deletes one item with the smallest ratio of the marginal gain on $f$ and cost $c$, until the solution no longer violates the new budget. However, AGGA cannot maintain an approximation for the new budget and may even perform arbitrarily badly~\cite{aij22Roostapour}.

\textbf{POMC.} Qian \textit{et al.}~\cite{POMC} proposed POMC, a Pareto Optimization method for maximizing a Monotone function with a monotone Cost constraint, which reformulates the original problem Eq.~(\refeq{eq:general-def}) as a bi-objective maximization problem that maximizes the objective function $f$ and minimizes the approximate cost function $\hat{c}$ simultaneously~\cite{ecj15submodular,qian2015subset}. To solve the bi-objective problem, POMC employs a simple multi-objective EA, i.e., GSEMO~\cite{LaumannsTEC04,qian2019}, which uses uniform selection and bit-wise mutation to generate an offspring solution and keeps the non-dominated solutions generated-so-far in the population. After terminated, it returns the best feasible solution with the largest $f$ value in the population. For the static setting of a fixed budget $B$, POMC can achieve the best known $(\alpha_f/2)(1-e^{-\alpha_f})$-approximation ratio, and also regain this approximation ratio at most $O(n\Delta_BP_{max}/\delta_{\hat{c}})$ expected running time when the budget changes~\cite{POMC,aij22Roostapour}. 


\textbf{EAMC.} Bian \textit{et al.}~\cite{EAMC} proposed a single-objective EA for maximizing
a Monotone function with a monotone Cost constraint, called EAMC. It tries to maximize a surrogate objective $g$ which considers both the original objective $f$ and the cost $\hat{c}$. For $|X|\ge 1$, $g(X)=f(X)/(1-e^{-\alpha_f \hat{c}(X)/B})$, while for $|X| = 0$, $g(X) = f(X)$. The submodularity ratio $\alpha_f$, used to calculate the surrogate objective $g$, may require exponential time to compute accurately, so a lower bound on $\alpha_f$ is often used instead. EAMC also applies uniform selection and bit-wise mutation to generate an offspring solution as POMC. For each subset size $i$, EAMC contains the solutions with the largest $g$ or $f$ values, bounding the maximum population size. After terminated, EAMC returns the feasible solution with the largest $f$ value in the population. EAMC can achieve the best-known $(\alpha_f/2)(1-e^{-\alpha_f})$-approximation in a static setting. However, in a dynamic setting, the $g$ function, based on the old budget, cannot characterize the new problem well, potentially leading to poor performance with the new budget.

\textbf{FPOMC.} Bian et al.~\cite{fpomc} then proposed the Fast Pareto Optimization algorithm for maximizing a Monotone function with a monotone Cost constraint, called FPOMC, which is modified from POMC. The main difference between FPOMC and POMC is that FPOMC applies a greedy selection strategy, while POMC applies uniform selection. The greedy selection strategy uses a function $h_{Z}(X)$ to estimate the goodness of a solution $X$ w.r.t. a reference point $Z$, which is defined as
\begin{equation}\nonumber
	h_{Z}(X)=\left\{\begin{array}{ll}
		(f(X)-f(Z))/(\hat{c}(X)-\hat{c}(Z)) &  \quad\hat{c}(X)>\hat{c}(Z), \\
		(f(X)-f(Z)) \cdot C +\hat{c}(Z)- \hat{c}(X) & \quad\hat{c}(X) \leq \hat{c}(Z),
	\end{array}\right.
\end{equation}
where $C$ is a large enough number. Intuitively, $h_{Z}(X)$ measures the goodness of $X$ by the marginal gain on $f$ and $c$  w.r.t. a reference point $Z$, and the solution with the largest $h$ value is selected with a high probability. For more detailed design of FPOMC, please refer to~\cite{fpomc}. For the static setting, FPOMC can obtain the best known $(\alpha_f/2)(1-e^{-\alpha_f})$-approximation ratio, and also regain the same approximation ratio at most $O(n K_{B'}(K_{B'}-K_{B}))$ expected running time for the new budget.

AGGA and the static GGA are greedy algorithms. POMC, EAMC and FPOMC are anytime algorithms that can find better solutions using more time. Among them, POMC performs the best empirically in dynamic environments~\cite{aij22Roostapour}; however, it was initially designed for static settings and lacks considerations for dynamic environments. This work focuses on designing an advanced algorithm tailored for dynamic environments, aiming to quickly adapt its solutions to budget changes and potentially exceed the speed of the static GGA.

\section{The BPODC Algorithm}\label{Algorithm}
In this section, we propose an algorithm based on Biased Pareto Optimization~\cite{aij22Roostapour,DynamicPartitionMatroid,0002Z022,ecj15submodular,Qian20,qian2021multiobjective,0001BF20,QianLFT23,QianLZ22,QianSYTZ17,QianS0TZ17,qian2019,qian2015subset,QianZT018,zhang2023sparsity} for maximizing a monotone function with Dynamic Cost constraints, called BPODC, which can quickly adapt its solutions when budget changes. It represents a subset $X \subseteq V$ by a Boolean vector $\bm{x} \in \{0,1\}^n$, where the $i$-th bit $x_i\!=\!1$ iff $v_i \in X$. We will not distinguish $\bm{x} \in \{0,1\}^n$ and its corresponding subset $\{v_i\in V | x_i = 1\}$ for notational convenience. BPODC reformulates the original problem Eq.~(\ref{eq:general-def}) as a bi-objective maximization problem
\begin{align}\label{def-CO-BO}
\arg\max\nolimits_{\bm{x} \in \{0,1\}^n}& \;  \big(f_1(\bm{x}),\;f_2(\bm{x})\big),
\end{align}
\begin{equation*}
\text{where } \begin{aligned}
f_1(\bm{x}) = \begin{cases}
	-\infty, & \hat{c}(\bm x)> B+1\\
	f(\bm x), &\text{otherwise}
\end{cases},\quad
f_2(\bm{x}) =-\hat{c}(\bm{x}).
\end{aligned}
\end{equation*} 
That is, BPODC maximizes the objective function $f$ and the negative of the approximate cost function $\hat{c}$ simultaneously. Solutions with cost values over $B+1$ (i.e., $\hat{c}(\bmx) > B+1$) are excluded by setting their $f_1$ values to $-\infty$. We use the value $B+1$ to give the algorithm a slight look ahead for larger constraint bounds without making the population size too large. The introduction of the second objective $f_2$ can naturally bring a diverse population, which may lead to better optimization performance.

Note that the objective vector $(f_1(\bmx), f_2(\bmx))$ is calculated only when the solution $\bmx$ is generated. This means that any subsequent changes to $B$ do not trigger an update of the objective vector. Thus, solutions exceeding cost $B'+1$ for a new budget $B'$ are still retained in the population. However, for any new solutions exceeding $B'+1$, the $f_1$ value is set to $-\infty$. As the two objectives may be conflicting, the domination relationship in Definition~\ref{def:domination} is often used for comparing two solutions. A solution is Pareto optimal if no other solution dominates it. The collection of objective vectors of all Pareto optimal solutions is called the Pareto front.

\begin{definition}[Domination]\label{def:domination}
For two solutions $\bm x$ and $\bm x'$, 
\begin{itemize}
	\item $\bm{x}$ \emph{weakly dominates} $\bm{x}'$, denoted as $\bm{x} \succeq \bm{x}'$, if $f_1(\bm{x}) \geq f_1(\bm{x}') \wedge f_2(\bm{x}) \geq f_2(\bm{x}')$; 
	\item $\bm{x}$ \emph{dominates} $\bm{x}'$, denoted as $\bm{x} \succ \bm{x}'$, if ${\bm{x}} \succeq \bm{x}'$ and either $f_1({\bm{x}}) > f_1(\bm{x}')$ or $f_2(\bm{x}) > f_2(\bm{x}')$; 
	\item they are \emph{incomparable} if neither $\bm{x} \succeq \bm{x}'$ nor $\bm{x}' \succeq \bm{x}$. 
\end{itemize}
\end{definition}

\begin{algorithm}[t]
\caption{BPODC Algorithm}
\label{alg:BPODC}
\textbf{Input}: a monotone objective function $f$, a monotone approximate cost function $\hat{c}$, a ground set with $n$ items, and a sequence of changes on the budget B\\
\textbf{Output}: a solution $\bm{x} \in \{0,1\}^n$ with $\hat{c}(\bm{x}) \leq B$\\
\textbf{Process}:
\begin{algorithmic}[1]
\STATE Let $\bm{x}=0^{n}$, $P=\{\bm{x}\}$;
\REPEAT
\IF{warm-up stage} 
\STATE Select $\bmx$ from $P$ uniformly at random
\ELSE
\STATE $\bmx=\text{Biased Selection}(P,B)$
\ENDIF
\STATE Generate $\bm{x}'$ by flipping each bit of $\bm{x}$ with probability $1/n$;
\IF {$\nexists \bm{z} \in P$ such that $\bm{z} \succ \bm{x}'$}
\STATE $P = (P \setminus  \{\bm{z} \in P \mid \bm{x}' \succeq \bm z\})\cup\{\bm{x}'\}$
\ENDIF
\UNTIL{some criterion is met}
\STATE \textbf{return} $\arg\max_{\bm x\in P, \hat{c}(\bm x)\leq B } f(\bm{x})$
\end{algorithmic}
\end{algorithm}

After constructing the bi-objective problem in Eq.~(\ref{def-CO-BO}), BPODC solves it by a process of multi-objective EAs, as described in Algorithm~\ref{alg:BPODC}. EAs, inspired by Darwin’s theory of evolution are general-purpose randomized heuristic optimization algorithms~\cite{back:96,DBLP:books/sp/ZhouYQ19}, mimicking variational reproduction and natural selection, which have become the most popular tool for multi-objective optimization~\cite{coello2007evolutionary,hong2021evolutionary,PanMWZJYH23}. It starts from the empty set $0^{n}$ (line~1), and repeatedly improves the quality of solutions in population $P$ (lines~2--12). At the start of the process, BPODC applies uniform selection to select a parent solution $\bm x$ for a period of time, which is referred to the warm-up stage (lines~3-4). Our aim is to uniformly explore the $(0, B+1]$ area initially to obtain a population with good diversity. After the warm-up stage, BPODC selects a parent solution $\bmx$ in $P$ according to the \textsc{Biased Selection} subroutine in Algorithm~\ref{alg:select} (line~6). Then, a solution $\bmx'$ is generated by applying bit-wise mutation on $\bmx$ (line~8), which is used to update the population $P$ (line~9-10). If $\bm{x}'$ is not dominated by any solution in $P$ (line~9), it will be added into $P$, and meanwhile, those solutions weakly dominated by $\bm{x}'$ will be deleted (line~10). This updating procedure makes the population $P$ always contain incomparable solutions. After running a number of iterations, the best feasible solution with the largest $f$ value in the population $P$ is output (line~13). Note that the aim of BPODC is to find a good solution of the original problem in Definition~\ref{def:dynamic-problem}, rather than the Pareto front of the reformulated bi-objective problem in Eq.~(\ref{def-CO-BO}). That is, the bi-objective reformulation is an intermediate process.

\begin{algorithm}[t]
	\caption{Biased Selection($P$, $B$): Subroutine of BPODC}
	\label{alg:select}
	\textbf{Input}: the population $P$ and the budget $B$\\
	\textbf{Output}: a solution in $P$ for mutation\\
	\textbf{Process}:
	\begin{algorithmic}[1]
        \STATE $\epsilon \gets 1 \times 10^{-10}$;
        \FOR{$i=1$ to $|P|$}
        \STATE $\bmx \gets$ the $i$-th solution in population $P$;
        \STATE $\textit{probs}[i] \gets 1/(|c(\bmx)-B|+\epsilon)$
        \ENDFOR
	\STATE $\textit{probs} \gets \textit{Normalization}(\textit{probs})$;
        \STATE  Select $\bmx$ from $P$ according to the probabilities $\textit{probs}$
	\RETURN $\bmx$ 
	\end{algorithmic}
\end{algorithm}

The \textsc{Biased Selection} subroutine in Algorithm~\ref{alg:select} first computes a selection probability of each solution $\bm x\in P$ iteratively, which is inversely proportional to the difference between the cost value $c(\bm x)$ and the given budget $B$ (lines~2-5). The $\epsilon$ is added to avoid division by zero (line~4). The subroutine then normalizes the probabilities, and selects a solution $\bm x$ from $P$ based on these normalized probabilities (lines~6-7). Algorithm~\ref{alg:select} exhibits a bias whereby solutions with cost values close to the budget $B$ have a higher probability of selection. This enables BPODC to quickly regain high-quality solutions upon budget changes, thereby meeting the demands of dynamic environments.

Note that BPODC uses the warm-up strategy only for the initial budget, starting from the zero solution; if the budget changes, it continues from the current population using biased selection.

\section{Theoretical Analysis}
In this section, we prove the general approximation bound of BPODC in Theorem~\ref{thm1}, implying that BPODC can achieve the best known $(\alpha_f/2)(1-e^{-\alpha_f})$-approximation guarantee for the static problem in Definition~\ref{def:static-problem}. When facing a dynamic budget change, BPODC still can regain the $(\alpha_f/2)(1-e^{-\alpha_f})$-approximation ratio (Theorem~\ref{thm2}). 

Let $p_{min}$ denote the minimum probability of selecting a solution from the population during the run of BPODC and  $\delta_{\hat{c}}=\min\{\hat{c}(X\cup v)-\hat{c}(X) | X\subseteq V, v\notin X\}$. We assume that $\delta_{\hat{c}}>0$. The proof of Theorem~\ref{thm1} is based on the approach used in Theorem~2 of~\cite{POMC}, mainly analyzing the expected number of iterations of BPODC required to obtain an $(\alpha_f/2)(1-e^{-\alpha_f})$-approximation solution.

\begin{theorem}\label{thm1}
	For the static problem in Definition~\ref{def:static-problem}, BPODC using at most $O(nB/(p_{min}\cdot\delta_{\hat{c}}))$ expected number of iterations finds a subset $X\subseteq V$ with
	$$
	f(X)\ge (\alpha_f/2)\cdot (1-e^{-\alpha_f})\cdot f(\tilde{X}),
	$$
	where $f(\tilde{X})$ is defined in Eq.~(\ref{eq-optimal}).
\end{theorem}
\begin{proof}
We follow the proof of POMC in Theorem~2 of~\cite{POMC}, which analyzes the increase of a quantity $J_{max}$ during the process of POMC, where $J_{max}=\max\{j\in [0,B) | \exists \bmx \in P, \hat{c}(\bm x)\leq j \land f(\bm x)\ge (1-(1-\alpha_f\frac{j}{Bk})^k)\cdot f(\tilde{X}) \text{ for some } k\}$. Let $\bmx \in P$ be a solution corresponding to the current value of $J_{max}$. It is pointed out that $J_{max}$ can increase at least $\delta_{\hat{c}}$ by adding a specific item to $\bmx$. Furthermore, they proved that POMC can find a solution with the $f$ value at least 
$$\max\{f(\bm{p}),f(\bm{q})\}\ge (\alpha_f/2)\cdot(1-e^{-\alpha_f})\cdot f(\tilde{X}),$$ 
where $\bm{p}$ is the solution that results from increasing $J_{max}$ at most $B/\delta_{\hat{c}}$ times starting from $J_{max}=0$, and $\bm{q}$ is defined as $\bm{q}=\arg\max_{v\in V: \hat{c}(\bm{q})\leq B} f(v)$.

Our proof finishes by analyzing the expected number of iterations until BPODC contains two solutions $\bm{p}$ and $\bm{q}$. We first analyze the expected number of iterations of BPODC to generate the solution $\bm{p}$. The initial value of $J_{max}$ is 0 as BPODC starts from $\{0\}^n$. Assume the current value of $J_{max}$ is $i$, and $\bmx\in P$ is a corresponding solution. According to the population update mechanism described in lines~9-10 of Algorithm~\ref{alg:BPODC}, $J_{max}$ cannot decrease because $\bm{x}$ can be deleted from $P$ (line~10) only when the newly included solution $\bm{x}'$ weakly dominates $\bm{x}$, i.e., $f(\bm{x}')\ge f(\bm{x})$ and $\hat{c}(\bm{x}') \leq \hat{c}(\bm{x})$, which makes $J_{max}\ge i$. As mentioned above, to increase $J_{max}$ from $i$ to at least $i+\delta_{\hat{c}}$, we can add a specific item to $\bm{x}$ to generate a new solution $\bm x'$. Upon generating $\bm x'$, it will be included into $P$; otherwise, $\bm x'$ must be dominated by one solution in $P$ (line~9 of Algorithm~\ref{alg:BPODC}), and this implies that $J_{max}$ has already been larger than $i$, which contradicts with the assumption $J_{max}=i$. In each iteration, the probability of successfully increasing $J_{max}$ is at least $p_{min}\cdot\frac{1}{n}\cdot(1-\frac{1}{n})^{n-1}\ge\frac{p_{min}}{en}$, where $p_{min}$ is a lower bound on the probability of selecting $\bmx$ in line~4 or line~6 of Algorithm~\ref{alg:BPODC} and $\frac{1}{n}\cdot(1-\frac{1}{n})^{n-1}$ is the probability of flipping a specific bit of $\bmx$ while keeping
other bits unchanged in line~8. Then, it needs at most $en/p_{min}$ expected number of iterations to increase $J_{max}$ by at least $\delta_{\hat{c}}$. After at most $enB/(p_{min}\cdot\delta_{\hat{c}}) $ expected number of iterations, $\bm{p}$ will be generated and included into $P$; otherwise, $P$ has already contained a solution $\bm{z} \succeq \bm{p}$, i.e., $\hat{c}(\bmz)\leq \hat{c}(\bm{p})\leq B$ and $f(\bmz)\ge f(\bm{p})$.

We now analyze the expected number of iterations to generate and contain the solution $\bm{q}$. Since $\{0\}^n$ has the largest $f_2$ value (i.e., the smallest $\hat{c}$ value), no other solution can dominate it, ensuring that $\{0\}^n$ will always be included in $P$. Thus, $\bm{q}$ can be generated in one iteration by selecting $\{0\}^n$ in line~4 or 6 of Algorithm~\ref{alg:BPODC} and flipping only the corresponding 0-bit in line~8, whose probability is at least $p_{min}/en$. That is, $\bm{q}$ will be generated using at most $en/p_{min}$ expected number of iterations.

Taking the expected number of iterations for generating $\bm{p}$ and $\bm{q}$ together, BPODC using at most $O(nB/(p_{min}\cdot\delta_{\hat{c}}))$ expected number of iterations finds a solution with the $f$ value at least $\max\{f(\bm{p}),f(\bm{q})\}\ge (\alpha_f/2)\cdot(1-e^{-\alpha_f})\cdot f(\tilde{X})$. 
\end{proof}

\begin{theorem}\label{thm2}
For the dynamic problem in Definition~\ref{def:dynamic-problem}, assume that BPODC has achieved an $(\alpha_f/2)(1-e^{-\alpha_f})$-approximation ratio for the current budget $B$ after running at most $O(nB/(p_{min}\cdot\delta_{\hat{c}}))$ expected number of iterations: 

(1) when $B$ decreases to $B'$, BPODC has already achieved the $(\alpha_f/2)(1-e^{-\alpha_f})$-approximation ratio for the new budget;  

(2) when $B$ increases to $B'$, BPODC using at most $O(n(B'-B)/(p_{min}\cdot\delta_{\hat{c}}))$ expected number of iterations, can regain the $(\alpha_f/2) (1-e^{-\alpha_f})$-approximation ratio.
\end{theorem}
\begin{proof}
By analyzing the increasing process of $J_{max}$ from 0 to $B$ as shown in the proof of Theorem~\ref{thm1}, BPODC using $O(nB/(p_{min}\cdot\delta_{\hat{c}}))$ expected number of iterations has achieved an $(\alpha_f/2)\cdot (1-e^{-\alpha_f})$-approximation ratio for the problem with any budget $B'\leq B$. When the budget $B$ increases to $B'$, in order to achieve the desired approximation guarantee, it is sufficient to continuously increase $J_{\max}$ from $B$ to $B'$. The expected number of iterations required for the increase is $O(n(B'-B)/(p_{min}\cdot\delta_{\hat{c}}))$. 
\end{proof}




\section{Empirical Study}
In this section, we empirically examine the performance of BPODC on dynamic variants of subset selection problems, specifically focusing on influence maximization and maximum coverage tasks where the cost constraint varies over time. We compare BPODC with several competitive algorithms: the static greedy algorithm GGA, the dynamic greedy algorithm AGGA, and the EA-based methods, POMC, EAMC, and FPOMC.

POMC and FPOMC use the same bi-objective as BPODC, as shown in Eq.~(\ref{def-CO-BO}), which is calculated only when the solution $\bmx$ is generated, and any subsequent changes to the budget $B$ do not trigger an update, that is, algorithms do not delete solutions from the population that become infeasible under a new budget. For EAMC, a budget change influences the value of surrogate function $g(X)=f(X)/(1-e^{-\alpha_f \hat{c}(X)/B})$. Thus, the value of function $g$ for solutions in the population is updated after each dynamic change. Note that the parameter $\alpha_f$ equals 1 because the objective functions of influence maximization and maximum coverage are submodular. The static greedy algorithm GGA requires $n(n+1)/2$ objective evaluations, denoted as $T_{G}$. For each budget change, the number of objective evaluations for EAs (BPODC, POMC, EAMC, and FPOMC) is set to $t = \{ 0.25T_{G},0.5T_{G}\}$, which is less than that required by GGA. Note that for the $t$ evaluations of the initial budget, BPODC uses the warm-up strategy to create a diverse population; for subsequent $t$ evaluations of a changed budget, it employs biased selection on the current population without a warm-up. For randomized algorithms (BPODC, POMC, EAMC and FPOMC), we independently run 30 times and report the average results. The source code is available at \href{https://github.com/lamda-bbo/BPODC}{https://github.com/lamda-bbo/BPODC}.


The experiments are mainly to answer two questions: Can BPODC perform the best among all algorithms under dynamic environments? Can BPODC adapt its solution within a shorter running time than the static GGA?

\textbf{Influence Maximization.} 
We use the same two datasets as social networks in influence maximization as in~\cite{EAMC}, called \textit{graph100} (100 vertices, 3,465 edges) and \textit{graph200} (200 vertices, 9,950 edges), respectively. The probability of each edge is set to 0.05. Besides, we use a larger real-world dataset, \textit{frb35-17-1} (595 vertices, 27,856 edges), with each edge's probability 0.01. We consider the linear cost constraint, where $c(X)=\sum_{v\in X}c_v$. The cost of each item is calculated based on its out-degree $d(v)$, i.e., $c_v=1+(1+|\xi|)\cdot d(v)$, where $\xi$ is a random number drawn from the normal distribution $\mathcal{N}(0,0.5^2)$. The original budget is set to 300, and stays within the interval $[100, 500]$. We consider a sequence of 100 budget changes obtained by randomly changing the current budget $B$ by a value of $[-10,10]$. The cumulative changes relative to the initial budget are depicted in Figure~\ref{fig-budgets}(a). The current budget at each time is calculated by adding the $y$-value at that point to the initial budget. To calculate $\mathbb{E}[|IC(X)|]$ in our experiments, we simulate the random propagation process starting from the solution $X$ for 500 times independently, and use the average as an estimation. Due to time constraints, we limit the number of objective evaluations of EAs (BPODC, POMC, EAMC, and FPOMC) on the \textit{frb35-17-1} dataset, i.e., $t_0 = 0.5T_{G}$ evaluations for the initial budget and $t = \{ 0.05T_{G},0.1T_{G}\}$ evaluations for upcoming budgets. This setting is to ensure that the study could be completed within the available timeframe while still providing meaningful insights into the algorithm's performance trends. Since the behavior of the greedy algorithm GGA or AGGA is randomized under noise, we also repeat its run 30 times independently and report the average results. The curves of average results over each time of change are plotted in Figure~\ref{fig-IM}, where the shaded areas indicate the standard deviation around the mean.

\begin{figure}[t!]
  \centering
  \begin{minipage}{0.4\textwidth} 
    \centering
    \includegraphics[width=0.8\linewidth]{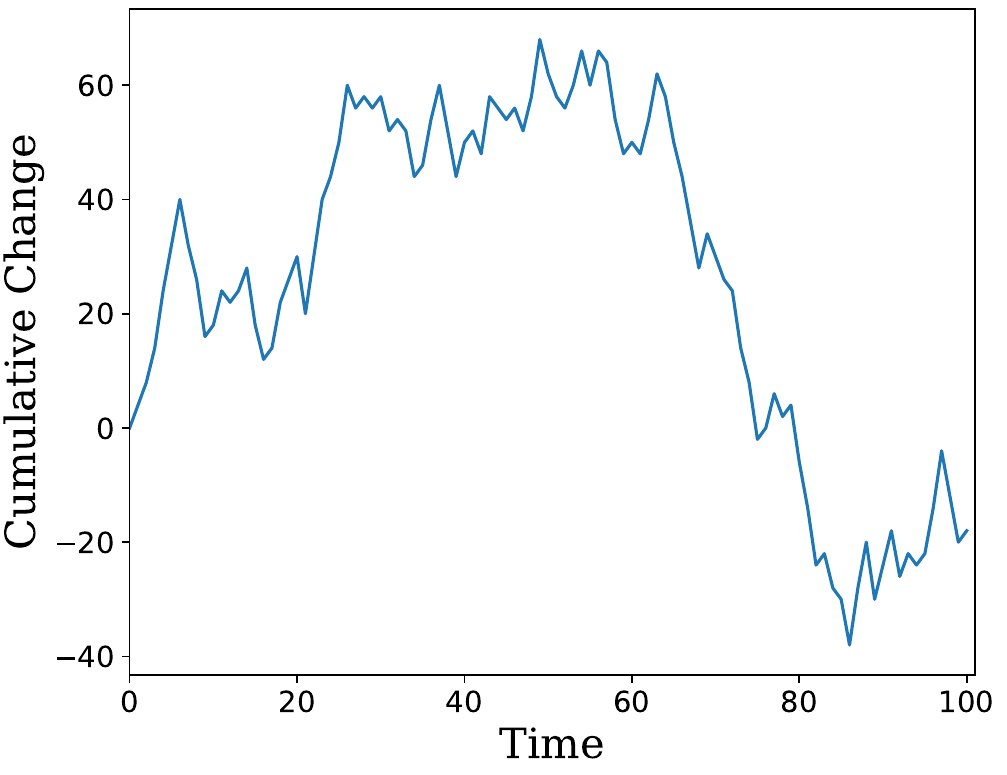}
  \end{minipage}
  \hspace{0.05\textwidth} 
  \begin{minipage}{0.4\textwidth} 
    \centering
    \includegraphics[width=0.8\linewidth]{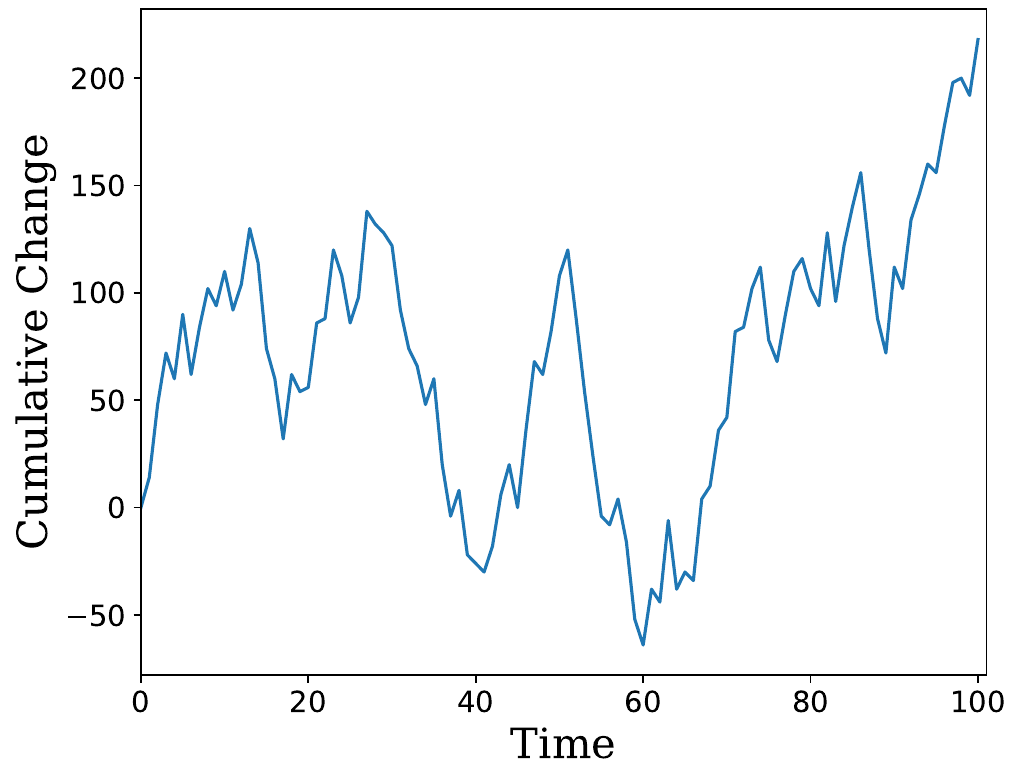}
  \end{minipage}\\
  \begin{minipage}[c]{0.5\linewidth}\centering
    \small(a) Influence Maximization
  \end{minipage}
  \begin{minipage}[c]{0.4\linewidth}\centering
    \small(b) Maximum Coverage
  \end{minipage}\vspace{-0.5em}
    \caption{The cumulative budget changes relative to the initial budget.}\label{fig-budgets}\vspace{-1em}
\end{figure}

\begin{figure}[t!]
\begin{minipage}[c]{1\linewidth}\centering
        \includegraphics[width=1\linewidth]{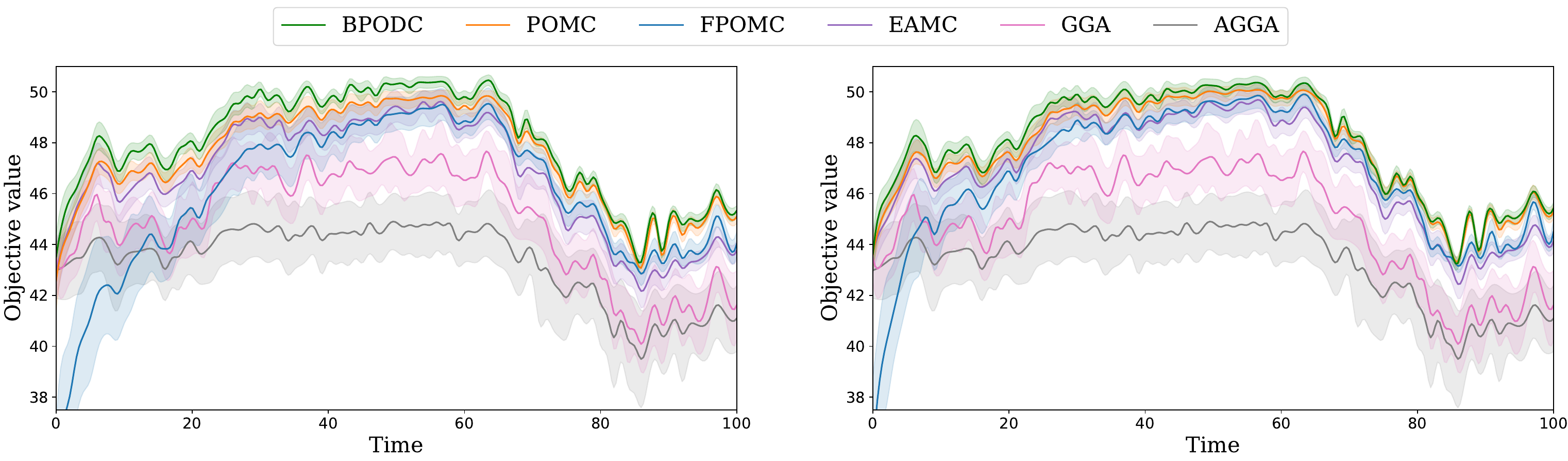}
\end{minipage}\\
\begin{minipage}[c]{0.49\linewidth}\centering
    \small(a) \textit{graph100}, $t=0.25T_G$
\end{minipage}
\begin{minipage}[c]{0.49\linewidth}\centering
    \small(b) \textit{graph100}, $t=0.5T_G$
\end{minipage}\\
\begin{minipage}[c]{1\linewidth}\centering
        \includegraphics[width=1\linewidth]{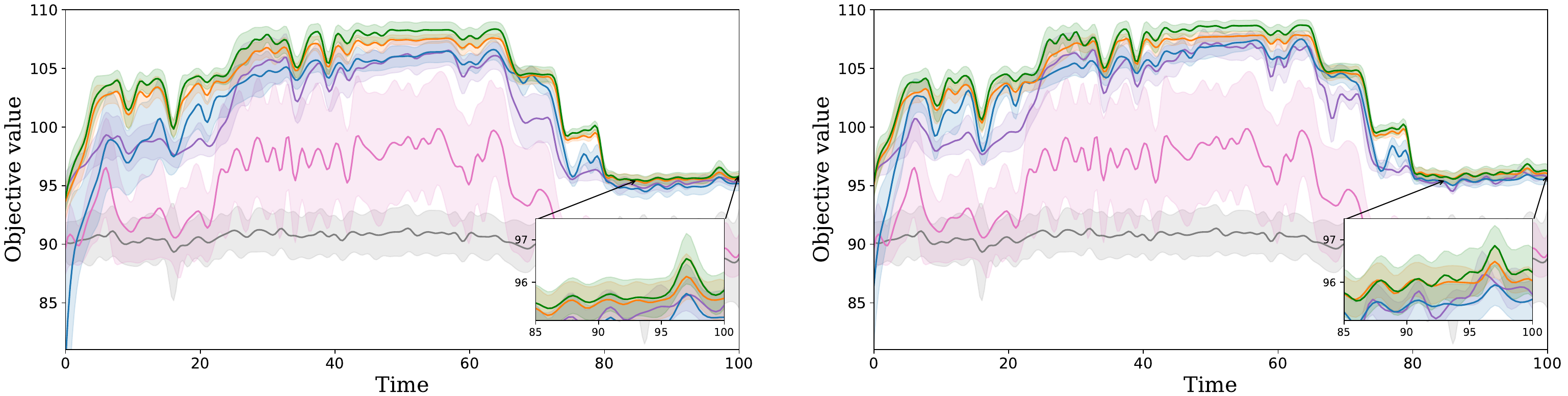}
\end{minipage}\\
\begin{minipage}[c]{0.49\linewidth}\centering
    \small(c) \textit{graph200}, $t=0.25T_G$
\end{minipage}
\begin{minipage}[c]{0.49\linewidth}\centering
    \small(d) \textit{graph200}, $t=0.5T_G$
\end{minipage}\\
\begin{minipage}[c]{1\linewidth}\centering
        \includegraphics[width=1\linewidth]{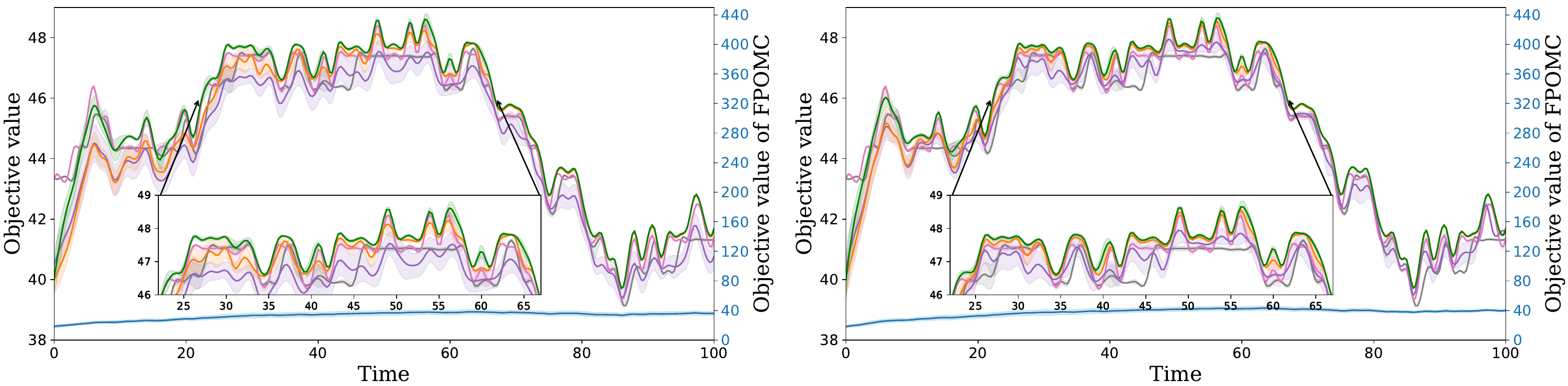}
\end{minipage}\\
\begin{minipage}[c]{0.49\linewidth}\centering
    \small(e) \textit{frb35-17-1}, $t_0=0.5T_G$ and $t=0.05T_G$
\end{minipage}
\begin{minipage}[c]{0.49\linewidth}\centering
    \small(f) \textit{frb35-17-1}, $t_0=0.5T_G$ and $t=0.1T_G$
\end{minipage}
\caption{The average value $\pm$ std for influence maximization under each budget change.}\label{fig-IM}\vspace{-1.5em}
\end{figure}

Figure~\ref{fig-IM}(a)-(d) clearly shows that BPODC, POMC and EAMC consistently outperform GGA and AGGA and are more stable, exhibiting a smaller std. FPOMC initially obtains a lower value during the first few changes, and then exceeds GGA and AGGA. These observations highlight the superior ability of EAs to leverage existing computational results while adapting to changes in budget constraints. Among all the algorithms, BPODC performs the best, regardless of the settings at $t=0.25T_G$ or $t=0.5T_G$. In Figure~\ref{fig-IM}(e)-(f), we set the running time allowed for EAs after each dynamic change to be much shorter, specifically $0.05T_G$ and $0.1T_G$, respectively. BPODC is the first to surpass GGA and AGGA and continues to consistently perform the best. Note that the curve of FPOMC is plotted on the secondary $y$-axis due to its poor performance.

\begin{figure}[t!]
\begin{minipage}[c]{1\linewidth}\centering
        \includegraphics[width=1\linewidth]{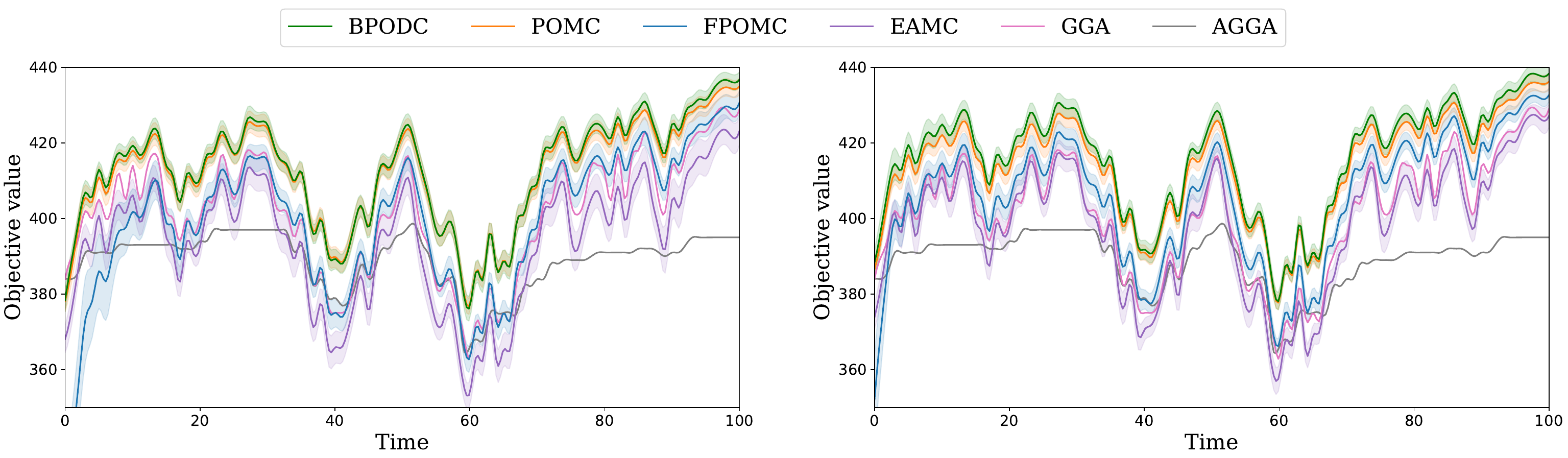}
\end{minipage}\\
\begin{minipage}[c]{0.49\linewidth}\centering
    \small(a) \textit{frb30-15-1}, $t=0.25T_G$
\end{minipage}
\begin{minipage}[c]{0.49\linewidth}\centering
    \small(b) \textit{frb30-15-1}, $t=0.5T_G$
\end{minipage}\\
\begin{minipage}[c]{1\linewidth}\centering
        \includegraphics[width=1\linewidth]{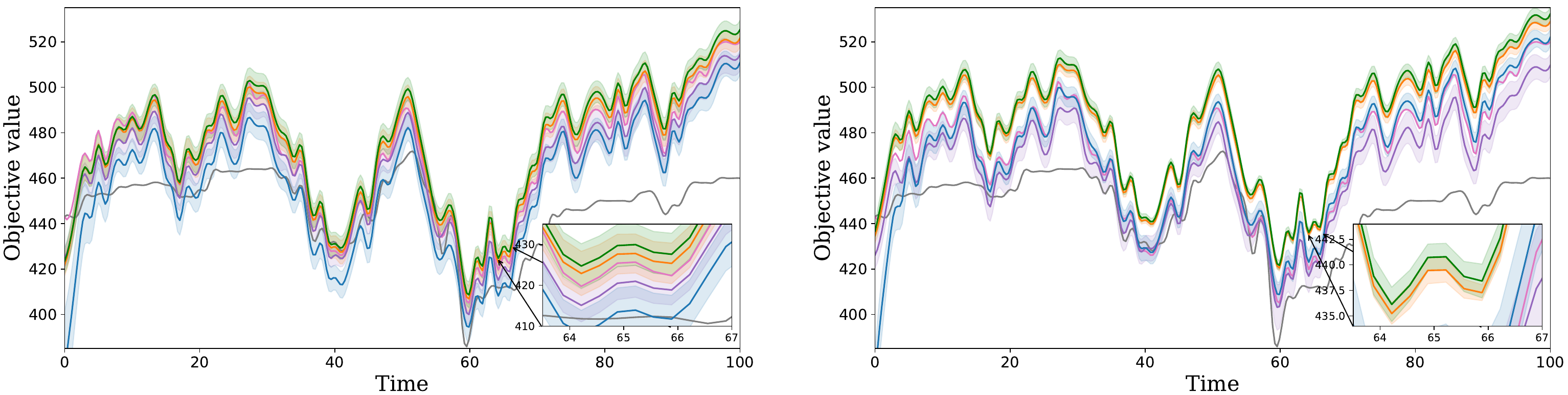}
\end{minipage}\\
\begin{minipage}[c]{0.49\linewidth}\centering
    \small(c) \textit{frb35-17-1}, $t=0.25T_G$
\end{minipage}
\begin{minipage}[c]{0.49\linewidth}\centering
    \small(d) \textit{frb35-17-1}, $t=0.5T_G$
\end{minipage}\\
\begin{minipage}[c]{1\linewidth}\centering
        \includegraphics[width=1\linewidth]{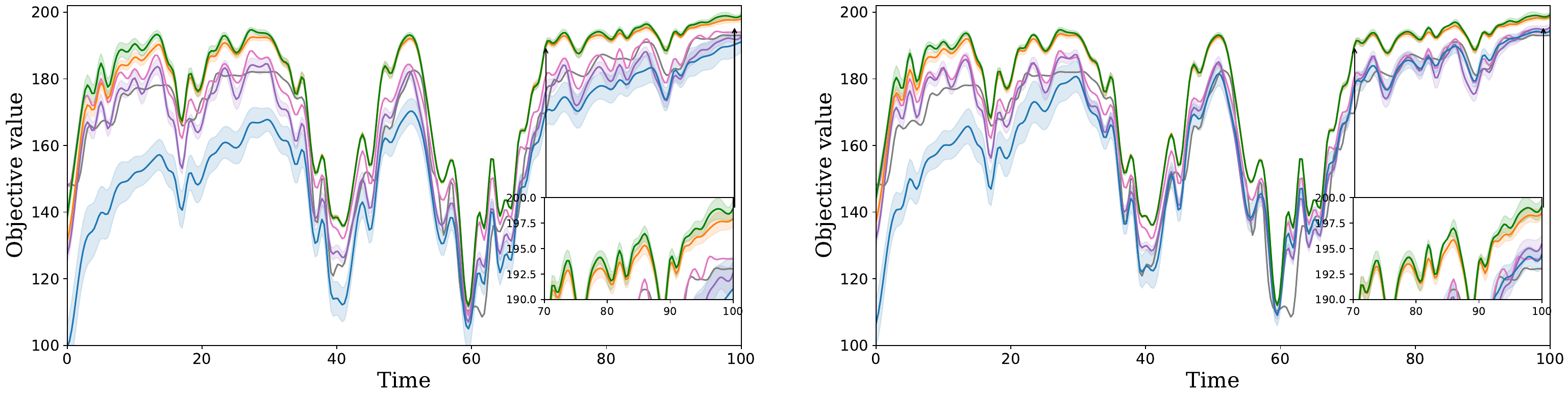}
\end{minipage}\\
\begin{minipage}[c]{0.49\linewidth}\centering
    \small(e) \textit{aves}, $t=0.25T_G$
\end{minipage}
\begin{minipage}[c]{0.49\linewidth}\centering
    \small(f) \textit{aves}, $t=0.5T_G$
\end{minipage}
\caption{The average value $\pm$ std for maximum coverage under each budget change.}\label{fig-MC}\vspace{-1.5em}
\end{figure}

\textbf{Maximum Coverage.} We use three real-world graph datasets for maximum coverage: \textit{frb30-15-1} (450 vertices, 17,827 edges) and \textit{frb35-17-1} (595 vertices, 27,856 edges), both used in~\cite{EAMC,aij22Roostapour}, and \textit{aves} (202 vertices, 4,658 edges)~\cite{wildbird}. For each vertex, we generate a set which contains the vertex itself and its adjacent vertices. We still use the linear cost constraint $c(X)=\sum_{v\in X}c_v$. The cost of each vertex is $c_v=1+\max\{d(v)-q,0\}$ as in~\cite{HarshawFWK19}, where $d(v)$ is the out-degree of vertex $v$ and $q$ is a constant (which is set to 6 in our experiment). For datasets \textit{frb30-15-1} and \textit{frb35-17-1}, the original budget is set to 500 and ranges from 300 to 700; for dataset \textit{aves}, it is set to 50 and ranges from 100 to 300. The sequence of 100 budget changes, randomly varying the current budget $B$ by $[-40,40]$, is shown in Figure~\ref{fig-budgets}(b) relative to the initial budget. The curves of average results over each time of change are plotted in Figure~\ref{fig-MC}.

\begin{figure}[t!]
\centering
\begin{minipage}[c]{0.38\linewidth}
\centering
        \includegraphics[width=0.9\linewidth]{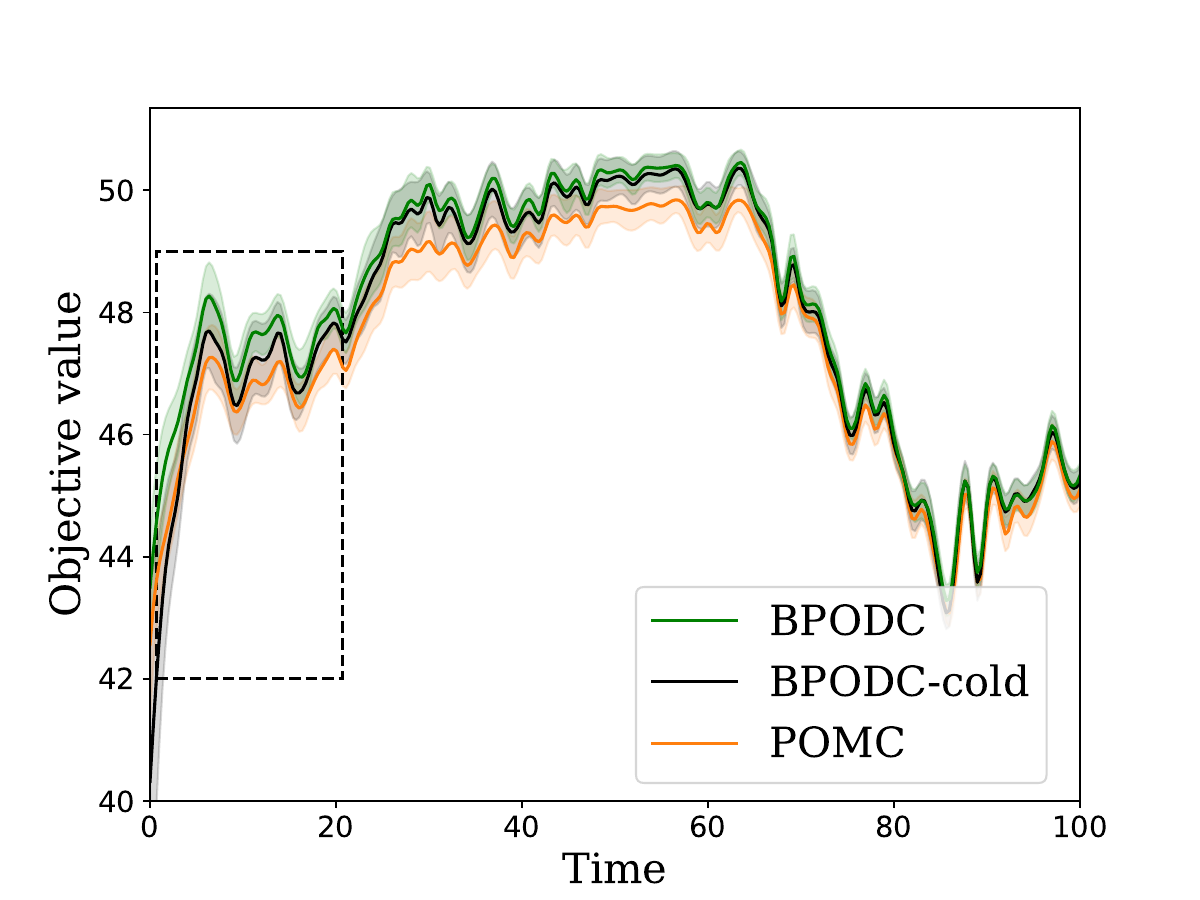}
\end{minipage}
\begin{minipage}[c]{0.4\linewidth}
\centering
        \includegraphics[width=0.9\linewidth]{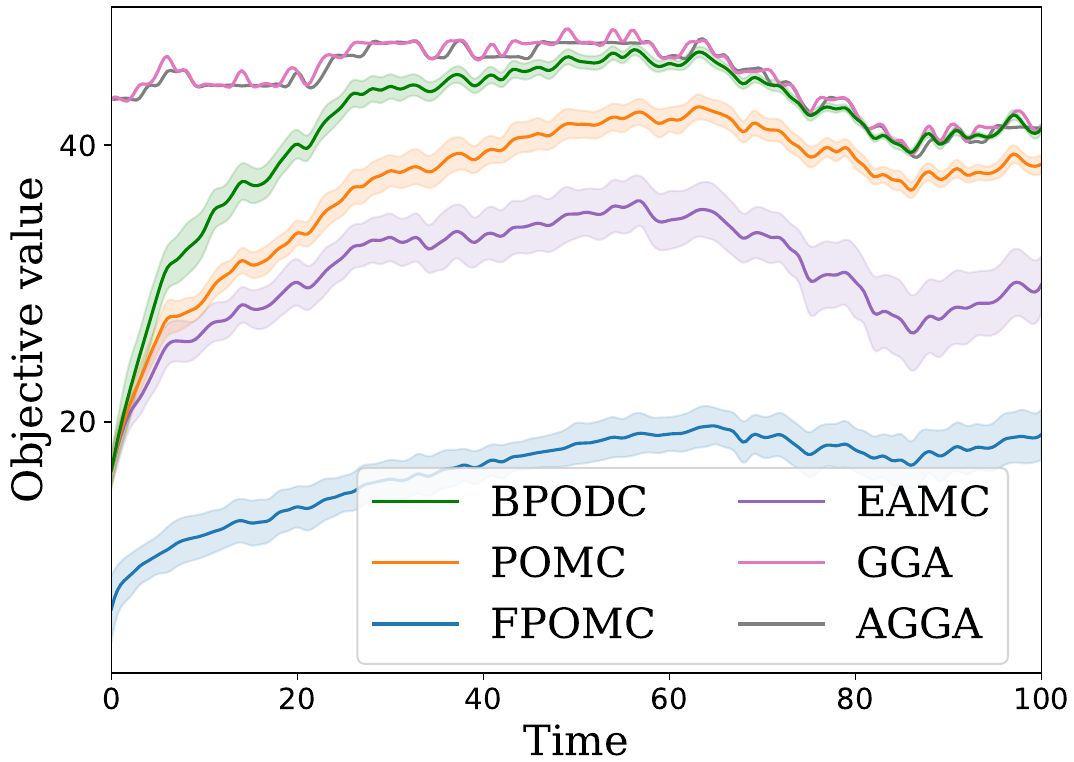}
\end{minipage}\\
\begin{minipage}[c]{0.4\linewidth}\centering
    \small(a) \textit{graph100}, $t=0.25T_G$
\end{minipage}
\begin{minipage}[c]{0.4\linewidth}\centering
    \small(b) \textit{frb35-17-1}, $t=1000$
\end{minipage}
\caption{Additional results for influence maximization under dynamic environment.}\label{fig-ablation}\vspace{-2em}
\end{figure}

Figure~\ref{fig-MC} shows that BPODC consistently performs the best across three datasets, although POMC sometimes matches its performance. However, the performance of other algorithms is less stable, for example, FPOMC performs moderately on \textit{frb30-15-1} (Figure~\ref{fig-MC}(a)-(b)) but poorly on \textit{aves} (Figure~\ref{fig-MC}(e)-(f)). EAMC underperforms GGA at both $t=0.25T_G$ and $t=0.5T_G$.

\textbf{Ablation Study.} As mentioned in Section~\ref{Algorithm}, the warm-up stage of BPODC initially applies uniform selection for a period, enhancing the effectiveness of subsequent biased selection. We test this by running BPODC-cold on the \textit{graph100} dataset with $t=0.25T_G$, using only biased selection without a warm-up. Figure~\ref{fig-ablation}(a) shows that BPODC-cold has a performance gap compared to BPODC during the first 20 changes, yet achieves similar performance afterward. Despite solely using biased selection, BPODC-cold still outperforms POMC. This implies the effectiveness of both biased selection and warm-up.

Next, we conduct experiments on the dataset \textit{frb35-17-1} for influence maximization in an extreme scenario. The number of objective evaluations for BPODC, POMC, EAMC, and FPOMC is set to 1000, including for the warm-up stage of BPODC. The results are plotted in Figure~\ref{fig-ablation}(b). As expected, the EAs do not surpass GGA, given that the available time resources are extremely limited—merely 0.006 times that of the greedy algorithm GGA. However, BPODC demonstrates a significantly faster adaptation speed compared to other EAs and maintains a performance level similar to GGA in the later stages (at 80th-100th changes).

\section{Conclusion}

This paper proposes BPODC, an enhancement of POMC that uses biased selection and warm-up strategies for subset selection with dynamic cost constraints. We prove that BPODC can maintain the best known $(\alpha_f/2)(1-e^{-\alpha_f})$-approximation guarantee when the budget changes. Experiments on influence maximization and maximum coverage show that BPODC adapts faster and more effectively to budget changes, utilizing a running time that is less than that of the greedy algorithm GGA. BPODC always achieves the best performance empirically. In the future, it is interesting to examine the performance of BPODC in more applications. 


\bibliographystyle{splncs04}
\bibliography{ref}

\begin{thebibliography}{10}
\providecommand{\url}[1]{\texttt{#1}}
\providecommand{\urlprefix}{URL }
\providecommand{\doi}[1]{https://doi.org/#1}

\bibitem{back:96}
B{\"{a}}ck, T.: Evolutionary Algorithms in Theory and Practice: Evolution Strategies, Evolutionary Programming, Genetic Algorithms. Oxford University Press, Oxford, UK (1996)

\bibitem{EAMC}
Bian, C., Feng, C., Qian, C., Yu, Y.: An efficient evolutionary algorithm for subset selection with general cost constraints. In: Proceedings of the 34th {AAAI} Conference on Artificial Intelligence (AAAI'20). pp. 3267--3274. New York, NY (2020)

\bibitem{fpomc}
Bian, C., Qian, C., Neumann, F., Yu, Y.: Fast {P}areto optimization for subset selection with dynamic cost constraints. In: Proceedings of the 30th International Joint Conference on Artificial Intelligence (IJCAI'21). pp. 2191--2197. Montreal, Canada (2021)

\bibitem{0002Z022}
Bian, C., Zhou, Y., Qian, C.: Robust subset selection by greedy and evolutionary {P}areto optimization. In: Proceedings of the 31st International Joint Conference on Artificial Intelligence (IJCAI'22). pp. 4726--4732. Vienna, Austria (2022)

\bibitem{dynamicgraphcoloring}
Bossek, J., Neumann, F., Peng, P., Sudholt, D.: Runtime analysis of randomized search heuristics for dynamic graph coloring. In: Proceedings of the 21st ACM Conference on Genetic and Evolutionary Computation Conference (GECCO'19). pp. 1443--1451. Prague, Czech Republic (2019)

\bibitem{coello2007evolutionary}
Coello, C.A.C., Lamont, G.B., van Veldhuizen, D.A.: Evolutionary Algorithms for Solving Multi-Objective Problems. Springer, New York, NY (2007)

\bibitem{DynamicPartitionMatroid}
Do, A.V., Neumann, F.: Pareto optimization for subset selection with dynamic partition matroid constraints. In: Proceedings of the 35th AAAI Conference on Artificial Intelligence (AAAI'21). pp. 12284--12292. Virtual (2021)

\bibitem{MC-application}
Feige, U.: A threshold of ln \emph{n} for approximating set cover. Journal of the {ACM}  \textbf{45}(4),  634--652 (1998)

\bibitem{ecj15submodular}
Friedrich, T., Neumann, F.: Maximizing submodular functions under matroid constraints by evolutionary algorithms. Evolutionary Computation  \textbf{23}(4),  543--558 (2015)

\bibitem{HarshawFWK19}
Harshaw, C., Feldman, M., Ward, J., Karbasi, A.: Submodular maximization beyond non-negativity: Guarantees, fast algorithms, and applications. In: Proceedings of the 36th International Conference on Machine Learning (ICML'19). pp. 2634--2643. Long Beach, California (2019)

\bibitem{hong2021evolutionary}
Hong, W.J., Yang, P., Tang, K.: Evolutionary computation for large-scale multi-objective optimization: A decade of progresses. International Journal of Automation and Computing  \textbf{18}(2),  155--169 (2021)

\bibitem{IM-application}
Kempe, D., Kleinberg, J., Tardos, {\'{E}}.: Maximizing the spread of influence through a social network. In: Proceedings of the 9th {ACM} {SIGKDD} International Conference on Knowledge Discovery and Data Mining (KDD'03). pp. 137--146. Washington, DC (2003)

\bibitem{SP-application}
Krause, A., Singh, A., Guestrin, C.: Near-optimal sensor placements in {G}aussian processes: Theory, efficient algorithms and empirical studies. Journal of Machine Learning Research  \textbf{9},  235--284 (2008)

\bibitem{LaumannsTEC04}
Laumanns, M., Thiele, L., Zitzler, E.: Running time analysis of multiobjective evolutionary algorithms on pseudo-{B}oolean functions. IEEE Transactions on Evolutionary Computation  \textbf{8}(2),  170--182 (2004)

\bibitem{nemhauser1978analysis}
Nemhauser, G.L., Wolsey, L.A., Fisher, M.L.: An analysis of approximations for maximizing submodular set functions {--} {I}. Mathematical Programming  \textbf{14}(1),  265--294 (1978)

\bibitem{Qian20}
Qian, C.: Distributed {P}areto optimization for large-scale noisy subset selection. IEEE Transactions on Evolutionary Computation  \textbf{24}(4),  694--707 (2020)

\bibitem{qian2021multiobjective}
Qian, C.: Multi-objective evolutionary algorithms are still good: Maximizing monotone approximately submodular minus modular functions. Evolutionary Computation  \textbf{29}(4),  463--490 (2021)

\bibitem{0001BF20}
Qian, C., Bian, C., Feng, C.: Subset selection by {P}areto optimization with recombination. In: Proceedings of the 34th AAAI Conference on Artificial Intelligence (AAAI'20). pp. 2408--2415. New York, NY (2020)

\bibitem{QianLFT23}
Qian, C., Liu, D., Feng, C., Tang, K.: Multi-objective evolutionary algorithms are generally good: Maximizing monotone submodular functions over sequences. Theoretical Computer Science  \textbf{943},  241--266 (2023)

\bibitem{QianLZ22}
Qian, C., Liu, D., Zhou, Z.: Result diversification by multi-objective evolutionary algorithms with theoretical guarantees. Artificial Intelligence  \textbf{309},  103737 (2022)

\bibitem{POMC}
Qian, C., Shi, J., Yu, Y., Tang, K.: On subset selection with general cost constraints. In: Proceedings of the 26th International Joint Conference on Artificial Intelligence (IJCAI'17). pp. 2613--2619. Melbourne, Australia (2017)

\bibitem{QianSYTZ17}
Qian, C., Shi, J., Yu, Y., Tang, K., Zhou, Z.: Optimizing ratio of monotone set functions. In: Proceedings of the 26th International Joint Conference on Artificial Intelligence (IJCAI'17). pp. 2606--2612. Melbourne, Australia (2017)

\bibitem{QianS0TZ17}
Qian, C., Shi, J., Yu, Y., Tang, K., Zhou, Z.: Subset selection under noise. In: Advances in Neural Information Processing Systems 30 (NeurIPS'17). pp. 3560--3570. Long Beach, CA (2017)

\bibitem{DBLP:conf/ijcai/Qian0T18}
Qian, C., Yu, Y., Tang, K.: Approximation guarantees of stochastic greedy algorithms for subset selection. In: Proceedings of the 27th International Joint Conference on Artificial Intelligence (IJCAI'18). pp. 1478--1484. Stockholm, Sweden (2018)

\bibitem{qian2019}
Qian, C., Yu, Y., Tang, K., Yao, X., Zhou, Z.: Maximizing submodular or monotone approximately submodular functions by multi-objective evolutionary algorithms. Artificial Intelligence  \textbf{275},  279--294 (2019)

\bibitem{qian2015subset}
Qian, C., Yu, Y., Zhou, Z.: Subset selection by {P}areto optimization. In: Advances in Neural Information Processing Systems 28 (NeurIPS'15). pp. 1765--1773. Montreal, Canada (2015)

\bibitem{QianZT018}
Qian, C., Zhang, Y., Tang, K., Yao, X.: On multiset selection with size constraints. In: Proceedings of the 32nd Conference on Artificial Intelligence (AAAI'18). pp. 1395--1402. New Orleans, LA (2018)

\bibitem{aij22Roostapour}
Roostapour, V., Neumann, A., Neumann, F., Friedrich, T.: Pareto optimization for subset selection with dynamic cost constraints. Artificial Intelligence  \textbf{302},  103597 (2022)

\bibitem{wildbird}
Rossi, R., Ahmed, N.: The network data repository with interactive graph analytics and visualization. In: Proceedings of the 29th {AAAI} Conference on Artificial Intelligence (AAAI'15). pp. 4292--4293. Austin,Texas (2015)

\bibitem{PanMWZJYH23}
Wu, T., Qian, H., Liu, Z., Zhou, J., Zhou, A.: Bi-objective evolutionary {B}ayesian network structure learning via skeleton constraint. Frontiers of Computer Science  \textbf{17}(6),  176350 (2023)

\bibitem{GGA}
Zhang, H., Vorobeychik, Y.: Submodular optimization with routing constraints. In: Proceedings of the 30th Conference on Artificial Intelligence (AAAI'16). pp. 819--826. Phoenix, AZ (2016)

\bibitem{zhang2023sparsity}
Zhang, L., Sun, X., Yang, H., Cheng, F.: Sparsity preserved {P}areto optimization for subset selection. IEEE Transactions on Evolutionary Computation  (2023)

\bibitem{DBLP:books/sp/ZhouYQ19}
Zhou, Z., Yu, Y., Qian, C.: Evolutionary Learning: Advances in Theories and Algorithms. Springer (2019)

\end{thebibliography}
%
%
%
%






\end{document}